\documentclass[runningheads]{llncs}

\usepackage{graphicx}
\usepackage{mathtools}
\usepackage{enumitem}
\usepackage{url}

\usepackage[T1]{fontenc}
\usepackage{microtype}

\usepackage{booktabs}
\usepackage{tabularx}

\usepackage{algorithm}
\usepackage[noend]{algpseudocode}
\newcommand{\algorithmicbreak}{\textbf{break}}
\newcommand{\Break}{\algorithmicbreak}

\usepackage{amsmath}
\newcommand{\pclass}{\ensuremath{\mathcal{P}}}
\newcommand{\jset}{\ensuremath{\mathcal{J}}}
\newcommand{\jhset}{\ensuremath{\mathcal{J}_\text{heavy}}}
\newcommand{\jmset}{\ensuremath{\mathcal{J}_\text{medium}}}
\newcommand{\jlset}{\ensuremath{\mathcal{J}_\text{light}}}
\newcommand{\jsset}{\ensuremath{\mathcal{J}_\text{scheduled}}}
\newcommand{\np}{\ensuremath{\mathcal{N\!P}}}
\newcommand{\pma}{\ensuremath{\text{P}}{}}
\newcommand{\qma}{\ensuremath{\text{Q}}{}}
\newcommand{\res}{\ensuremath{\text{\emph{res}}}{}}
\newcommand{\cmax}{\ensuremath{C_\text{max}}{}}
\newcommand{\opt}{\ensuremath{\text{\emph{OPT}}}{}}

\usepackage[caption=false]{subfig}

\title{A log-linear $(2+5/6)$-approximation algorithm for parallel machine scheduling with a single orthogonal resource\thanks{Preprint of the paper accepted at the 27th International European Conference on Parallel and Distributed Computing (Euro-Par 2021), Lisbon,
Portugal, 2021, DOI: \texttt{10.1007/978-3-030-85665-6\_10}}}
\titlerunning{A log-linear $(2+5/6)$-approximation algorithm...}

\author{Adrian Naruszko \and
Bartłomiej Przybylski \and
Krzysztof Rządca
}

\authorrunning{A. Naruszko et al.}

\institute{%
Institute of Informatics, University of Warsaw\\
Warsaw, Poland\\
\email{an371233@students.mimuw.edu.pl}, \email{\{bap,krzadca\}@mimuw.edu.pl}
}

\begin{document}

\maketitle

\begin{abstract}
As the gap between compute and I/O performance tends to grow, modern High-Performance Computing (HPC) architectures include a new resource type: an intermediate persistent fast memory layer, called burst buffers.
This is just one of many kinds of renewable resources which are orthogonal to the processors themselves, such as network bandwidth or software licenses. Ignoring orthogonal resources while making scheduling decisions just for processors may lead to unplanned delays of jobs of which resource requirements cannot be immediately satisfied.
We focus on a classic problem of makespan minimization for parallel-machine scheduling of independent sequential jobs with additional requirements on the amount of a single renewable orthogonal resource. We present an easily-implementable log-linear algorithm that we prove is $2\frac56$-approximation. In simulation experiments, we compare our algorithm to standard greedy list-scheduling heuristics and show that, compared to LPT, resource-based algorithms generate significantly shorter schedules.

\end{abstract}

\section{Introduction}

In a simplified model of parallel-machine scheduling, independent jobs can be freely assigned to processors if only at most one job is executed by each processor at any moment. Thus, the only resource to be allocated is the set of processors. However, this simple model does not fully reflect the real-life challenges. In practice, jobs may require additional resources to be successfully executed. These resources may include---among others---fast off-node memory, power, software licenses, or network bandwidth. All of these example resources are renewable, i.e. once a job completes it returns the claimed resources to the common pool (in contrast to consumable resources such as time or fuel).

Some of the resources may be orthogonal which means that they are allocated independently of other resources. For example, in standard High-Performance Computing (HPC) scheduling, node memory is not orthogonal as a job claims all the memory of the node on which it runs. On the other hand, in cloud computing, node memory is an orthogonal resource, as it is partitioned among containers or virtual machines concurrently running at the node. Thus, node memory (and also network bandwidth) is managed in a similar manner to the processors.

The results presented in this paper are directly inspired by the practical problem of parallel-machine scheduling of jobs in HPC centers with an orthogonal resource in a form of a burst buffer. A burst buffer is a fast persistent NVRAM memory layer logically placed between the operational memory of a node and the slow external file system. Burst buffers are shared by all the jobs running in a cluster and thus they are orthogonal to processors. They can be used as a cache for I/O operations or as a storage for intermediate results in scientific workflows.

The main contribution of this paper is a log-linear $2\frac56$-approximation algorithm for a parallel-machine scheduling problem with independent, sequential and rigid jobs, a single orthogonal resource and makespan as the objective. We thus directly improve the classic $\left(3 - \frac3m\right)$-approximation algorithm by Garey and Graham \cite{Garey1975}. Although a $(2 + \varepsilon)$-approximation algorithm \cite{NiemeierWiese2013} and an asymptotic FPTAS \cite{Jansen2019} are known, their time and implementation complexities are considerable. Our algorithm can be easily implemented and it runs in log-linear time. Additionally, it can be combined with known fast heuristics, thus providing good average-case results with a guarantee on the worst case.

The paper is organized as follows. In Sec.~\ref{sec:pd}, we define the analysed problem and review the related work. Then, in Sec.~\ref{sec:aa} we present a $2\frac56$-approximation algorithm and prove its correctness. In Sec.~\ref{sec:se} we simulate our algorithm, compare its performance to known heuristics, and discuss its possible extensions. Finally, in Sec.~\ref{sec:co}, we make general conclusions.

\section{Problem definition and related work}
\label{sec:pd}

We are given a set of $m$ parallel identical machines, a set of $n$ non-preemptable jobs $\jset = \{1, 2, \dots, n\},$ and a single resource of integer capacity $R.$ Each job $i \in \jset$ is described by its processing time $p_i$ and a required amount of the resource $R_i \leq R.$ We use the classic rigid job model \cite{Feitelson97} in which the processing time does not depend on the amount of the resource assigned. Our aim is to minimize the time needed to process all the jobs (or, in other words, the maximum completion time). Based on the three-field notation introduced in \cite{Graham1979} and then extended in \cite{Blazewicz1983}, we can denote the considered problem as $\pma|\res 1{\cdot}{\cdot}|\cmax.$ Here, the three values after the `\res' keyword determine: (1) the number of resources (one in this case); (2) the total amount of each resource (arbitrary in this  case); (3) the maximum resource requirement of a single job (arbitrary in this case). This problem is \np-Hard as a generalization of $\pma||\cmax$ \cite{GareyJohnson1978}. However, both the problems have been deeply analyzed in the literature. Moreover, different machine, job, and resource characteristics have been considered in the context of various objective functions. We refer the reader to \cite{Edis2013} for the most recent review on resource-constrained scheduling.

In case of the variant without additional resources -- $\pma||\cmax$ -- it was proved that the LPT strategy leads to a $\left(\frac43 - \frac1{3m}\right)$-approximation algorithm \cite{Graham1969} while any arbitrary list strategy provides a $\left(2 - \frac1m\right)$-approximation \cite{Graham1966}. A classic result on polynomial-time approximation scheme (PTAS) for $\pma||\cmax$ was presented in \cite{HochbaumSchmoys1987}. In general, an efficient polynomial-time approximation scheme (EPTAS) for the $\qma||\cmax$ problem (with uniform machines) is known \cite{Jansen2010,Jansen2016}. As the $\pma||\cmax$ problem is strongly \np-Hard, there exists no fully polynomial-time approximation scheme unless $\pclass=\np.$

When orthogonal resources are introduced, the upper-bounds increase. Given $s$ additional resources and more than two machines, any arbitrary list strategy leads to a $\left(s + 2 - \frac{2s + 1}m\right)$-approximation algorithm \cite{Garey1975} and this general bound is tight for $m > 2s + 1.$ For $s = 1,$ i.e. in the case of a single resource, this ratio becomes $3 - \frac3m.$ We also get $\lim_{m \to \infty}\left(3 - \frac3m\right) = 3.$

For the $\pma|\res 1{\cdot}1, p_i = 1|\cmax$ problem with unit processing times, binary resource requirements and an arbitrary amount of the resource, the optimal \cmax{} can be found in constant time \cite{Kovalyov1998}. The more general $\pma|\res 1{\cdot}1, r_i, p_i = 1|\cmax$ problem with ready times can be solved in linear time \cite{Blazewicz1978}, while the $\pma2|\res 1{\cdot}{\cdot}, r_i, p_i = 1|\cmax$ problem with just two machines and no limits on resource requirements of a single job is already \np-Hard \cite{Blazewicz1986}.
 
A number of heuristic and approximation algorithms has been developed for the $\pma|\res 1{\cdot}{\cdot}|\cmax$ problem and its close variants. A polynomial-time $\frac43$-approximation algorithm for the $\pma|\res 1{\cdot}{\cdot}, p_i = 1|\cmax$ problem was shown in \cite{Krause1973}. On the other hand, a $(3.5 + \varepsilon)$-approximation algorithm is known for the $\pma|\res 1{\cdot}{\cdot}, \emph{Int}|\cmax$ problem \cite{Kellerer2008}, i.e. for resource-dependent job processing times. In \cite{NiemeierWiese2013}, a $(2 + \varepsilon)$-approximation algorithm for the $\pma|\res 1{\cdot}{\cdot}|\cmax$ problem is presented. This algorithm can be transformed into a PTAS if the number of machines or the number of different resource requirements is upper-bounded by a constant. Further, an asymptotic FPTAS for the $\pma|\res 1{\cdot}{\cdot}|\cmax$ problem was shown \cite{Jansen2019}. Although the latter results have a great theoretical impact on the problem considered in this paper, the complexity of the obtained algorithms prevents them from being efficiently implemented.

\section{Approximation algorithm}
\label{sec:aa}

In this section, we present a log-linear $2\frac56$-approximation algorithm for the $\pma|\res 1{\cdot}{\cdot}|\cmax$ problem. In order to make our reasoning easier to follow, we normalize all the resource requirements. In particular, for each job $i \in \jset$ we use its relative resource consumption $r_i \coloneqq R_i/R \in [0, 1].$ Thus, for any set of jobs $J \subseteq \jset$ executed at the same moment it must hold that $\sum_{i \in J} r_i \leq 1.$ We also denote the length of an optimal schedule with \opt.

Before we present algorithm details, we introduce some definitions and facts. We will say that job $i \in \jset$ is \emph{light} if $r_i \leq \frac13.$ \emph{medium} if $\frac13 < r_i \leq \frac12,$ and \emph{heavy} if $r_i > \frac12.$ Thus, each job falls into exactly one of the three disjoint sets, denoted by $\jlset,$ $\jmset,$ $\jhset,$ respectively. Note that --- whatever the job resource requirements are --- not more than one heavy and one medium, or two medium jobs can be executed simultaneously.

Given a subset of jobs $J \subseteq \jset,$ we define its total resource consumption as $R(J) \coloneqq \sum_{i \in J} r_i.$ We say that a set $J$ is \emph{satisfied} with $\theta$ resources, if for each subset $J' \subseteq J$ such that $|J'| \leq m$ we have $R(J') \leq \theta.$ In such a case, we write $R_m(J) \leq \theta.$

The set of all the scheduled jobs will be denoted by $\jsset.$ For a scheduled job $i \in \jsset,$ we denote its start and completion times by $S_i$ and $C_i,$ respectively. We say that job $i$ was \emph{executed} before moment $t$ if $C_i \leq t,$ is being executed at moment $t$ if $S_i \leq t < C_i,$ and will be executed after moment $t$ if $S_i > t.$ The set of jobs executed at moment $t$ will be denoted by $\jset(t)$ and the total resource consumption of jobs in set $\jset(t)$ by $R(t).$ We present a brief summary of the notation in Tab.~\ref{table:notation}.

\begin{table}
\caption{Summary of the notation}
\label{table:notation}
\begin{tabularx}{\textwidth}{lX}
\toprule
Symbol & Meaning\\
\midrule
$\jset$ & The set of all the jobs\\
$\jset(t)$ & The set of all scheduled jobs $i$ such that $S_i \leq t < C_i$\\
$J \subseteq \jset$ & Subset of all the jobs\\
$R(J)$ & Total resource requirement of the jobs in $J$\\
$R_m(J)$ & Total resource requirement of $m$ most consumable jobs in $J$\\
$R(t)$ & Total resource requirement of jobs in $\jset(t)$\\
\bottomrule
\end{tabularx}
\end{table}

Let us start with the following lemmas, which are the essence of our further reasoning. In the algorithm, we will make sure that the assumptions of these lemmas are met, so we can use them to prove the general properties of the solution.

\begin{lemma}
\label{lem:rta}
If $R(t) \geq \frac23$ for all $t \in [t_a, t_b),$ then $t_b - t_a \leq \frac32\opt.$
\end{lemma}

\begin{proof}
It holds that $\opt \geq \sum_{i \in \jset} r_i\cdot p_i,$ as the optimal schedule length is lower-bounded by a total amount of resources consumed by all jobs in time. As $R(t) \geq \frac23,$ we obtain $$\frac23\cdot\left(t_b - t_a\right) \leq \int_{t_a}^{t_b}R(t) \,\textup{d}t \leq \sum_{i \in \jset} r_i\cdot p_i \leq \opt,$$ and thus $t_b - t_a \leq \frac32\opt.$
\end{proof}

\begin{lemma}
\label{lem:rtb}
Let $J_1, J_2 \subseteq \jset.$ If $R_m(J_1) < \theta_1$ and $R_m(J_2) < \theta_2,$ then $R_m(J_1 \cup J_2) < \theta_1 + \theta_2.$
\end{lemma}

\begin{proof}
Let us notice that, given any set of jobs $J,$ the value of $R_m(J)$ is determined by $m$ most resource-consuming jobs in $J.$ Thus, the value of $R_{m}(J_1 \cup J_2)$ is determined by $m$ most resource-consuming jobs from sets $J_1$ and $J_2.$ As a consequence, it holds that $R_{m}(J_1 \cup J_2) \leq R_m(J_1) + R_m(J_2) < \theta_1 + \theta_2.$
\end{proof}

\subsection{The idea}

Our approximation algorithm consists of four separate steps that lead to a resulting schedule $T.$ In each of the steps we generate a part of a schedule denoted by $T_1.$ $T_2.$ $T_3$ and $T_4,$ respectively. A general structure of schedule $T$ is presented in Fig.~\ref{fig:schedule-T}. Note that some time moments are marked by $t_1$, $t_2$, $t_3$ and $t_4$. These values are found while the algorithm is being executed. However, our algorithm guarantees that $0 \leq t_2 \leq t_1 \leq t_3 \leq t_4.$ We also state that $t_3 \leq \frac32\opt$ and that $t_4 - t_3 \leq \frac43\opt.$ As it is so, we get $\cmax(T) \leq  \frac32\opt + \frac43\opt = 2\frac56\opt.$ The algorithm is structured as follows.
\begin{enumerate}[label={\textbf{Step \arabic*.}}, leftmargin=1.4cm]
    \item Schedule all heavy jobs on the first machine in the weakly decreasing order of their resource requirements. As a consequence, it holds that $t_1 = \sum_{i \in \jhset} p_i$ and thus $t_1 \leq \opt.$ 
    \item Schedule selected light jobs starting from moment $0$ in such a way that $t_2 \leq t_1$ and $R(t) \geq \frac23$ for all $t < t_2.$
    \item Schedule all medium jobs and selected not-yet scheduled light jobs starting from moment $t_2$ in such a way that $t_3 \leq \frac32\opt.$
    \item Schedule all the remaining jobs using an LPT list strategy, starting from moment $t_3.$ In Steps (2) and (3), we selected jobs to be scheduled in such a way that now all non-scheduled jobs form a set $J$ such that $R_m(J) \leq 1.$ Thus, $t_4 - t_3 \leq \frac43\opt.$
\end{enumerate}

\begin{figure}[t]
    \includegraphics[width=\textwidth]{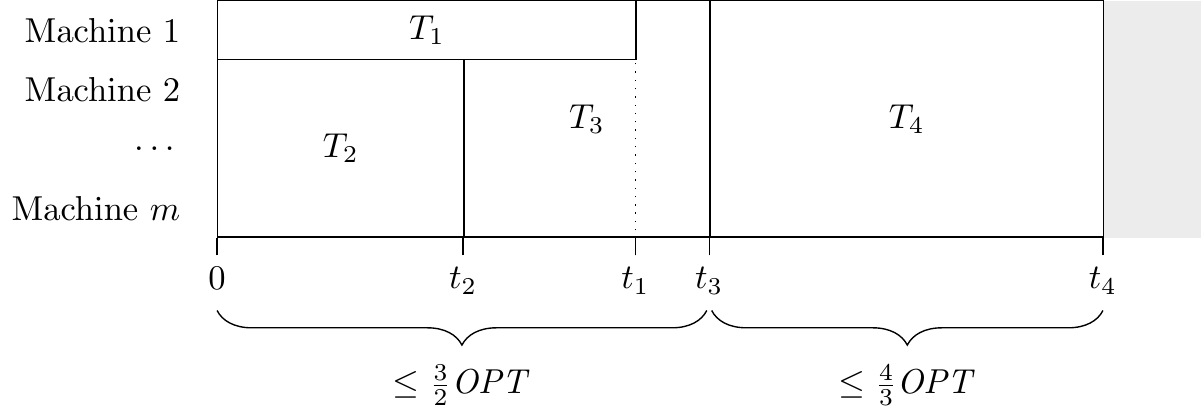}
    \caption{The structure of a resulting schedule $T.$} \label{fig:schedule-T}
\end{figure}

As Step~(1) is self-explanatory, we will not discuss it in details. However, we present its pseudocode in Alg.~\ref{alg:step-1}. Steps~(2) and (3) share a subroutine \textsc{Schedule-2/3} presented in Alg.~\ref{alg:schedule-23}. Given a set $J$ of jobs to be scheduled and a starting time $t_s,$ the procedure schedules some (perhaps none) of the jobs from $J$ and returns $t_c$ such that $R(t) \geq \frac23$ for all $t \in [t_s, t_c).$ If $t_s = t_c$, then $[t_s, t_c) = \emptyset$ and the statement remains true. Note that the time complexity of the \textsc{Schedule-2/3} subroutine is $O\left(|J|\log|J|\right).$

\begingroup
\setcounter{algorithm}{0}
\renewcommand\thealgorithm{\arabic{algorithm}A} 
\begin{algorithm}[h]
    \caption{Approximation algorithm --- Step 1 out of 4}
    \label{alg:step-1}
    \begin{algorithmic}
    \State {\textsc{Schedule} all the jobs from the $\jhset$ set in a weakly decreasing order of $r_i,$ on the first machine, starting from moment $0$}
    \State $t_1 \gets \sum_{i \in \jhset} p_i$
    \end{algorithmic}
\end{algorithm}
\endgroup

\begin{algorithm}[tb]
    \caption{}
    \label{alg:schedule-23}
    \begin{algorithmic}
    \Procedure{Schedule-2/3}{$J,$ $t_s$}
        \State {$t_c \gets t_s$}
        \State {\textsc{Sort} $J$ in a weakly decreasing order of $r_i$}

        \For {$i \in J$}
            \State {$t_c \gets \min\{ t \colon t \geq t_c \text{ and } R(t) < \frac23 \}$}
            \If {job $i$ can be feasibly scheduled in the $[t_c, t_c + p_i)$ interval}
                \State {\textsc{Schedule} job $i$ in the $[t_c, t_c + p_i)$ interval on any free machine}
            \Else {}
                \Break
            \EndIf
        \EndFor

        \State \Return $\min\{ t \colon t \geq t_c \text{ and } R(t) < \frac23 \}$
    \EndProcedure
    \end{algorithmic}
\end{algorithm}

\subsection{The analysis of the algorithm}

In Step~(2) of the algorithm, we partially fill the $T_2$ block of the schedule in such a way that at least $\frac23$ of the resource is consumed at any point in the $[0, t_2)$ interval. As we expect $t_2$ to be less or equal to $t_1,$ this step does not apply if $t_1 = 0$ or, equivalently, if $\jhset = \emptyset.$ The procedure itself is presented in Alg.~\ref{alg:step-2} and its time complexity is $O(|\jlset|\log|\jlset|).$ Note that when it is finished, there might exist one or more scheduled jobs $i \in \jlset$ such that $C_i > t_2,$ i.e. $\jset(t_2) \cap \jlset$ might be a non-empty set.

\begingroup
\setcounter{algorithm}{0}
\renewcommand\thealgorithm{\arabic{algorithm}B} 
\begin{algorithm}[tb]
    \caption{Approximation algorithm --- Step 2 out of 4}
    \label{alg:step-2}
    \begin{algorithmic}
    \State $t_2 \gets 0$
    \If {$\jhset \neq \emptyset$}
        \State $t_c \gets \textsc{Schedule-2/3}(\jlset, 0)$
        \State $t_2 \gets \min\{t_1, t_c\}$
        \If {$t_1 = t_2$}
            \State {\textsc{Unschedule} jobs $i \in \jlset$ for which $S_i \geq t_2$}
        \EndIf
    \EndIf
    \end{algorithmic}
\end{algorithm}
\endgroup

\begin{proposition}
\label{prop:step-2}
Let $t_1 > 0.$ After Step~(2) is finished, the following statements hold:
\begin{enumerate}[label=(\alph*)]
    \item If $t_1 = t_2,$ then $R(\jset(t_2) \cap \jlset) < \frac13.$
    \item If $t_1 > t_2,$ then $R(\jset(t_2) \cap \jlset) < \frac16$ and $R_m\left(\jlset \setminus \jsset\right) < \frac13.$
\end{enumerate}
\end{proposition}

\begin{proof}[1a]
Assume that $t_1 = t_2$ and let $h \in \jhset$ be the last heavy job scheduled in Step~(1). Note that, in particular, $C_h = t_1$. If $\jset(t_2) \cap \jlset = \emptyset,$ then obviously $R(\jset(t_2) \cap \jlset) = 0 < \frac13.$ Otherwise, each job in the $\jset(t_2) \cap \jlset$ set must have been started before $t_2.$ Consider a moment $t = t_2 - \varepsilon,$ for an arbitrarily small $\varepsilon.$ The set $\jset(t)$ consists of job $h,$ zero or more light jobs $i \in \jlset$ for which $C_i = t_2,$ and all the jobs from the $\jset(t_2) \cap \jlset$ set.

Select job $j \in \jset(t_2) \cap \jlset$ that was scheduled as the last one. Be reminded that in Step~(2), light jobs are scheduled in the decreasing order of their resource requirements.
From the construction of the algorithm we conclude that $R(t) - r_j < \frac23.$ If $r_j \geq \frac16,$ then it must be the only job in the $\jset(t_2) \cap \jlset$ set, as $r_h + r_j > \frac23.$ Thus, $\frac16 \leq R(\jset(t_2) \cap \jlset) < \frac13.$ If $r_j < \frac16,$ then $R(t) - r_j - r_h < \frac23 - \frac12 = \frac16,$ so $R(t) - r_h < \frac16 + r_j < \frac13.$ Thus, $R(\jset(t_2) \cap \jlset) < \frac13.$ \qed
\end{proof}

\begin{proof}[1b]
Assume that $t_1 > t_2$ and let $h \in \jhset$ be the heavy job executed at the moment $t_2.$ The value returned by the \textsc{Schedule-2/3} subroutine guarantees that $R(t) \geq \frac23$ for $0 \leq t < t_2.$ At the same time, it must hold that $\frac12 < R(t_2) < \frac23.$ Thus, at moment $t_2$ we have $R(\jset(t_2) \cap \jlset) + r_h < \frac23$ and, as a consequence, $R(\jset(t_2) \cap \jlset) < \frac 16.$

Now, we will show that $R_m\left(\jlset \setminus \jsset\right) < \frac13.$ There are two cases to be considered: either all the machines are busy at the moment $t_2$ or not. If not all the machines are busy, then all the light jobs were scheduled. Otherwise, any remaining light job would be scheduled on a free machine before the \textsc{Schedule-2/3} subroutine would end. Thus, $\jlset \setminus \jsset = \emptyset$ and $R_m\left(\jlset \setminus \jsset\right) = 0.$

Now, assume that all the machines are busy at the moment $t_2,$ i.e. $|\jset(t_2) \cap \jlset| = m - 1.$ As the light jobs were scheduled in the weakly decreasing order of their resource requirements and $R(\jset(t_2) \cap \jlset) < \frac16,$ which was proven before, we have
$$\max\{r_i \colon i \in \jlset \setminus \jsset\} \leq \min\{r_i \colon i \in \jset(t_2) \cap \jlset\} < \frac16.$$
As $R(\jset(t_2) \cap \jlset) < \frac16$ and $|\jset(t_2) \cap \jlset| < m,$ we conclude that $R_{m-1}(\jset(t_2) \cap \jlset) = R(\jset(t_2) \cap \jlset) < \frac16$ and $R_{m-1}\left(\jlset \setminus \jsset\right) \leq R_{m-1}(\jset(t_2) \cap \jlset) < \frac16.$
Finally,
$$\begin{array}{rcl}
    R_m\left(\jlset \setminus \jsset\right) & \leq & R_{m-1}\left(\jlset \setminus \jsset\right)\\
    & & + \max\{r_i \colon i \in \jlset \setminus \jsset\}\\
    & < & \frac16 + \frac16 = \frac13. \quad\qed
\end{array}$$
\end{proof}

Before Step~(3) is started, all heavy jobs and some (perhaps none) light jobs are scheduled. At this point, we ignore all the jobs from the $\jset(t_2) \cap \jlset$ set, i.e. the scheduled light jobs for which $S_i \leq t_2 < C_i.$ Being ignored, they are treated as scheduled jobs that do not occupy machines and do not use any resources. Thus, we will not schedule these jobs in Step~(3). These jobs will be rescheduled in Step~(4). As a consequence, at any point $t \geq t_2$ not more than a single heavy job is actually executed. In Step~(3), we first schedule all the medium jobs using a standard list scheduling approach, and then, if $t_1 = t_2$, we try to schedule not-yet scheduled light jobs using the \textsc{Schedule-2/3} routine. This intuition is formalized in Alg.~\ref{alg:step-3}. Note that the time complexity of Alg.~\ref{alg:step-3} is $O(|\jset|\log|\jset|).$

\begingroup
\setcounter{algorithm}{0}
\renewcommand\thealgorithm{\arabic{algorithm}C} 
\begin{algorithm}[tb]
    \caption{Approximation algorithm --- Step 3 out of 4}
    \label{alg:step-3}
    \begin{algorithmic}
    \State {\textsc{Ignore} all the jobs from the $\jset(t_2) \cap \jlset$ set}
    \State {Use a standard list scheduling approach to \textsc{schedule} all the jobs from $\jmset$ in a weakly increasing order of $r_i,$ starting from moment $t_2$}
    \State $t_g \gets \sup\{t \colon R(t) \geq \frac23\}$
    \If {$t_1 = t_2$}
        \State $t_g \gets \textsc{Schedule-2/3}(\jlset \setminus \jsset, t_g)$
    \EndIf
    \State $t_c \gets \max\{C_i \colon i \in \jsset\}$
    \State $t_3 \gets \max\{t_g, t_1\}$
    \end{algorithmic}
\end{algorithm}
\endgroup

\begin{proposition}
\label{prop:step-3}    
Let $t_c$ and $t_g$ be defined as in Alg.~\ref{alg:step-3}. After Step~(3) is finished, the following statements hold:
\begin{enumerate}[label=(\alph*)]
    \item If $t_c = t_1,$ then $R_m(\jset \setminus \jsset) < \frac13.$
    \item If $t_c > t_1 = t_2,$ then $R_m\left(\jset(t_g) \cup (\jset \setminus \jsset)\right) < \frac23.$
    \item If $t_c > t_1 > t_2,$ then $R_m\left(\jset(t_g) \cup (\jset \setminus \jsset)\right) < \frac56.$
\end{enumerate}
\end{proposition}

\begin{proof}[2a]
If $t_c = t_1$ and $t_1 = t_2,$ then no jobs were scheduled in Step~(3) and thus $|\jset \setminus \jsset| = 0.$ On the other hand, if $t_c = t_1$ and $t_1 > t_2,$ then all medium jobs are finished before or at $t_1.$ The only jobs that were not scheduled yet are light jobs. According to Prop.~\ref{prop:step-2}(b), we had $R_m\left(\jlset \setminus \jsset\right) < \frac13$ after Step~(2), so now it must hold that $R_m(\jset \setminus \jsset) < \frac13.$ \qed
\end{proof}

\begin{proof}[2b]
Notice that $R(t_g) < \frac23$ and thus at most one medium job is being executed at $t_g.$ If no medium jobs are being executed at $t_g$, then either there were no medium jobs to be scheduled or light jobs made the value of $t_g$ increase. In both cases, as $t_c > t_1 = t_2$, there are only light jobs being executed at $t_g$ and only light jobs are left to be scheduled. Thus, $R_m\left(\jset(t_g) \cup (\jset \setminus \jsset)\right) < \frac23$. Otherwise, the value of $t_g$ would be even larger.
Notice that if exactly one medium job is being executed at $t_g$, then this medium job is the last one executed. Consider two cases. If not all the machines are busy at $t_g,$ then there are no light jobs left to be scheduled and thus $\jset \setminus \jsset = \emptyset$ and $R_m\left(\jset(t_g) \cup (\jset \setminus \jsset)\right) = R(t_g) < \frac23.$ On the other hand, if all the machines are busy at $t_g,$ then a medium job and $m - 1$ light jobs are executed at this moment. As the medium job has larger resource requirement than any light job, and light jobs were scheduled in a weakly decreasing order of their resource requirements, it must hold again that $R_m\left(\jset(t_g) \cup (\jset \setminus \jsset)\right) = R(t_g) < \frac23.$ \qed
\end{proof}

\begin{proof}[2c]
From the assumption that $t_c > t_1 > t_2,$ we conclude that at least one medium job was scheduled in Step~(3) and $t_3 \geq t_g > \max_{i \in \jsset} S_i \geq t_2.$ As all the light jobs for which $S_i \geq t_2$ were unscheduled in Step~(2), all the light jobs for which $C_i > t_2$ were ignored, and $t_1 \neq t_2,$ no light jobs are executed at moment $t_g.$ If two non-light jobs, $i$ and $j,$ were executed at $t_g,$ then it would hold that $r_i + r_j > \frac23$ which contradicts the fact that $t_g \geq \sup\{t \colon R(t) \geq \frac23\}.$ Finally, if a heavy job was executed at $t_g,$ then it would hold that $t_c = t_1$ which contradicts the assumption that $t_c > t_1.$ Thus, at most one medium job can be executed at the moment $t_g,$ and $R(t_g) \leq \frac12.$ At the same moment, according to Prop.~\ref{prop:step-2}(b), we had $R_m\left(\jlset \setminus \jsset\right) < \frac13$ after Step~(2). Based on Lem.~\ref{lem:rtb}, we obtain $R_m\left(\jset(t_g) \cup (\jset \setminus \jsset)\right) < \frac12 + \frac13 = \frac56.$ \qed
\end{proof}

\begin{proposition}
It holds that $t_3 \leq \frac32\opt.$
\end{proposition}

\begin{proof}
It holds that $t_3 \geq t_1$. If $t_1 = t_3,$ then $t_3 \leq \opt \leq \frac32\opt,$ so assume that $t_3 > t_1.$ First, consider a case when $t_3 > t_1 = t_2$. It means that all medium jobs (if they exist) were scheduled starting from moment $t_1.$ As for any medium job $i \in \jmset$ we have $\frac13 < r_i \leq \frac12,$ any two medium jobs can be executed in parallel, and if it is so, more than $\frac23$ of the resource is used. In such a case, just after medium jobs are scheduled, we have $t_g = \sup\{t \colon R(t) \geq \frac23\},$ and after Step~(3) is finished, one has $R(t) \geq \frac23$ for all $t \in [t_2, t_3)$. Now, be reminded about the ignored jobs from the $\jset(t_2) \cap \jlset$ set. If we reconsider them, even at the cost of exceeding the available amount of resources or the number of machines, we have $R(t) \geq \frac23$ for all $t \in [0, t_3)$. This is enough to state that, based on Lem.~\ref{lem:rta}, $t_3 \leq \frac32\opt,$ although some of the jobs scheduled before $t_3$ are ignored and will be rescheduled later.

Now, consider a case when $t_3 > t_1 > t_2$. As $t_1 > t_2$, no light jobs are scheduled in Step~(3). This inequality implicates that at least one medium job could have been scheduled together with a heavy job in the $[t_2, t_1)$ interval. As heavy jobs are scheduled in a weakly decreasing order of the resource requirements, and medium jobs are scheduled in a weakly increasing order of the resource requirements, there are two possibilities.

If a medium job is being executed at every moment $t \in [t_2, t_1)$, then $t_g = \sup\{t \colon R(t) \geq \frac23\}$ and --- based on the same reasoning as in the previous case --- we have $R(t) \geq \frac23$ for all $t \in [0, t_g)$. As $t_3 = t_g$, based on Lem.~\ref{lem:rta} we obtain $t_3 \leq \frac32\opt.$

In the second possibility, there exists a point $t$ in the $[t_2, t_1)$ interval, for which $R(t) < \frac23.$ Consider the latest such point, i.e. $t \coloneqq \sup\left\{t \in [t_2, t_1) \colon R(t) < \frac23\right\}.$ The $t$ value is either equal to $t_1$, or to a starting time of a medium job. In both cases, no medium job $i$ for which $S_i \geq t$ could have been scheduled earlier. Thus, in the $[t, t_g)$ interval (if non-empty) all jobs are either heavy or medium, and are scheduled on exactly $2$ machines at the same time. Moreover, it would be not possible to execute three such jobs in parallel due to their resource requirements, so in the optimal schedule not more than two machines would be busy starting from point $t$ due to the resource requirements of the medium jobs. In our case, both machines are busy in the $[t, t_g)$ interval, so it must hold that $t_g \leq \opt$. If so, then $t_3 = \max\{t_g, t_1\} \leq \opt \leq \frac32\opt.$ \qed
\end{proof}

After Step~(3) is finished, we unschedule ignored jobs from the $\jset(t_2) \cap \jlset$ set and all the jobs from the $\jset(t_g)$ set, and then we schedule all jobs that are in the updated $\jset \setminus \jsset$ set. As we now know that $t_3 \leq \frac32\opt$ and that $t_3 \geq t_g$, all the machines are free starting from the $t_3$ moment. This intuition is shown in Alg.~\ref{alg:step-4}. It can be now shown that all the jobs to be scheduled are satisfied with a single unit of the resource. As it is so, all the machines can execute $m$ jobs in parallel, whichever $m$ jobs we choose. Thus, a schedule provided by the LPT list strategy is a $\frac43$-approximation solution for the $\jset \setminus \jsset$ set.

\begingroup
\setcounter{algorithm}{0}
\renewcommand\thealgorithm{\arabic{algorithm}D} 
\begin{algorithm}[tb]
    \caption{Approximation algorithm --- Step 4 out of 4}
    \label{alg:step-4}
    \begin{algorithmic}
    \State {\textsc{Unschedule} all the ignored jobs from the $\jset(t_2) \cap \jlset$ set}
    \State {\textsc{Unschedule} all the jobs from the $\jset(t_g)$ set}
    \State {Use an LPT list scheduling approach to \textsc{schedule} all the jobs from the  $\jset \setminus \jsset$ set (including just unscheduled ones), starting from moment $t_3$}
    \State $t_4 \gets \max\{t_3, \cmax\}$
    \end{algorithmic}
\end{algorithm}
\endgroup

\begin{proposition}
After Step~(3) is finished, it holds that $$R_m((\jset(t_2) \cap \jlset) \cup \jset(t_g) \cup (\jset \setminus \jsset)) < 1.$$
\end{proposition}

\begin{proof}
The proof follows directly from Prop.~\ref{prop:step-2}--\ref{prop:step-3} and Lem.~\ref{lem:rtb}.
\end{proof}

\begin{theorem}
Algorithm~\ref{alg:step-1}--\ref{alg:step-4} is a log-linear $2\frac56$-approximation algorithm for the $\pma|\res 1{\cdot}{\cdot}|\cmax$ problem.
\end{theorem}

Notice that the $2\frac56$-approximation ratio leaves us room for immediate improvement. In fact, based on the result by Graham \cite{Graham1969}, as we can apply the LPT strategy in Step (4) without any concerns about the orthogonal resource, it must hold that $t_4 - t_3 \leq \left(\frac43 - \frac1{3m}\right)\cdot\opt$. Thus, Alg.~\ref{alg:step-1}--\ref{alg:step-4} is $\left(2\frac56 - \frac1{3m}\right)$-approximation.

\section{Simulations and extensions}
\label{sec:se}

The log-linear algorithm presented in Sec.~\ref{sec:aa} provides a schedule that is not more than $2\frac56$ times longer than the optimal one. This is so for arbitrary independent jobs and arbitrary resource requirements.
While our principal contribution is in theory, our log-linear algorithm is also easily-implementable,
so in this section we evaluate our algorithm and compare its average-case performance to standard heuristics.

\subsection{Compared algorithms}

We will compare three variants of our algorithm against a number of greedy heuristics based on list scheduling algorithms. All the algorithms were implemented in \textsc{Python}; we performed our experiments on an Intel Core i7-4500U CPU @ 3.00GHz with 8GB RAM.

The theoretical algorithm from Sec.~\ref{sec:aa} will be denoted by \emph{ApAlg}. Its first extension, denoted by \emph{ApAlg-S}, introduces an additional step of backfilling. Namely, after the \emph{ApAlg} algorithm is finished, we iterate over all the jobs in order of their starting times, and reschedule them so they start at the earliest moment possible. Note that this additional step never increases the starting time of any job, and thus \emph{ApAlg-S} is also a $2\frac56$-approximation algorithm. The second extension, denoted by \emph{ApAlg-H} is a heuristic algorithm based on \emph{ApAlg}. In this case, the \textsc{Schedule-2/3} subroutine (see Alg.~\ref{alg:schedule-23}) is replaced. In \emph{ApAlg-H}, it schedules jobs in a strict weakly decreasing order of $r_i$ (so no job $j$ such that $r_j < r_i$ is started before job $i$), regardless of whether the total resource consumption at $t$ has exceeded $\frac23$.

We compare the \emph{ApAlg}, \emph{ApAlg-S} and \emph{ApAlg-H} algorithms against four list scheduling algorithms: \emph{LPT} (Longest Processing Time), \emph{HRR} (Highest Resource Requirement), \emph{LRR} (Lowest Resource Requirement), and \emph{RAND} (Random Order).
Any list scheduler starts by sorting jobs according to the chosen criterion.
Then, when making a scheduling decision, it seeks for the first job on the list that can be successfully scheduled, i.e. has its resource requirements not greater than what is left given jobs being currently executed (if no such job exists, or all processors are busy, the algorithm moves to the next time moment when any job completes).
Thus, the worst-case running time of the \emph{ApAlg-S}, \emph{LPT} and \emph{RAND} algorithms is $O(n^2)$. The worst-case running time of the \emph{LRR} and \emph{HRR} algorithms is $O(n\log n)$ -- these algorithms can use binary search to find the first job from the list having resource requirement not exceeding the currently available resources.

\subsection{Instances}

Our simulations are based on the dataset provided by the MetaCentrum Czech National Grid \cite{metacentrumlogs,Klusacek2017}. In order to avoid normalizing data from different clusters, we arbitrarily chose the cluster with the largest number of nodes (\emph{Zapat}). The \emph{Zapat} cluster consisted of $112$ nodes, each equipped with $16$ CPU cores and $134$GB of RAM. This cluster was monitored between January 2013 and April 2015, which resulted in $299\,628$ log entries.

Each entry provides information about job processing times ($p_i$) and their main memory requirements ($r_i$). We limit ourselves to entries for which both $p_i$ and $r_i$ are less or equal to their respective $99$th percentiles, so the data can be safely normalized. As different jobs may be executed in parallel on each node, we consider the main memory to be a single orthogonal resource. We assume that the total memory size (total amount of the resource) is equal to the maximum memory requirement in the set of all the considered job entries.
Thus, we normalize all the resource requirements so $r_i \in [0, 1]$ (where 1 is the resource capacity of the simulated system). As we study the problem with sequential jobs, we also assume that each job from the trace requires a single CPU. In Fig.~\ref{fig:prdist}, we present how the job processing times and normalized memory requirements were distributed within the log. Most of the jobs have rather low resource requirements. In fact, the $25$th, $50$th and $75$th percentiles of the resource requirements distribution are equal to $0.0165$, $0.0169$, and $0.068$, respectively. We have also analysed how the values of $p_i$ and $r_i$ correlate to each other. As the distributions of $p_i$ and $r_i$ clearly are not Gaussian-like, we calculated the Spearman's correlation coefficient. This requires an additional assumption that the relation between $p_i$ and $r_i$ is monotonic. The calculated value is $0.14207$ which suggests positive, yet not very significant correlation. This result was verified visually.

\begin{figure}[t]
    \vspace*{-1.5em}
    \subfloat{
        \includegraphics[width=0.48\textwidth]{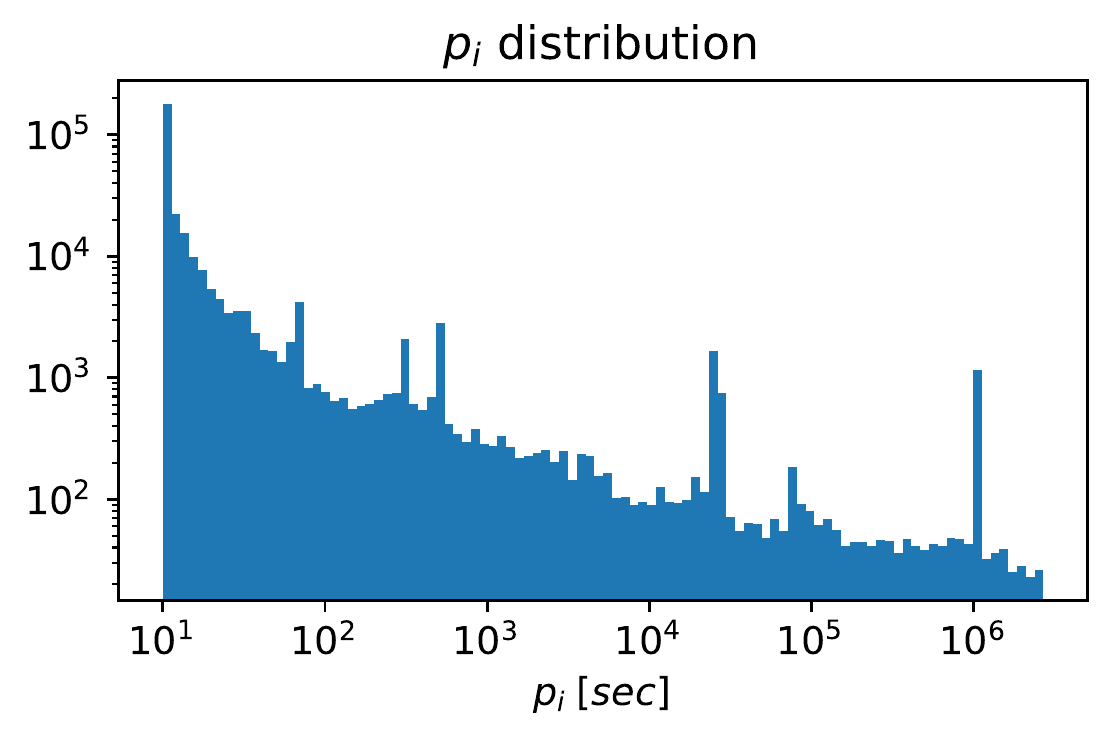}
    }
    \subfloat{
        \includegraphics[width=0.48\textwidth]{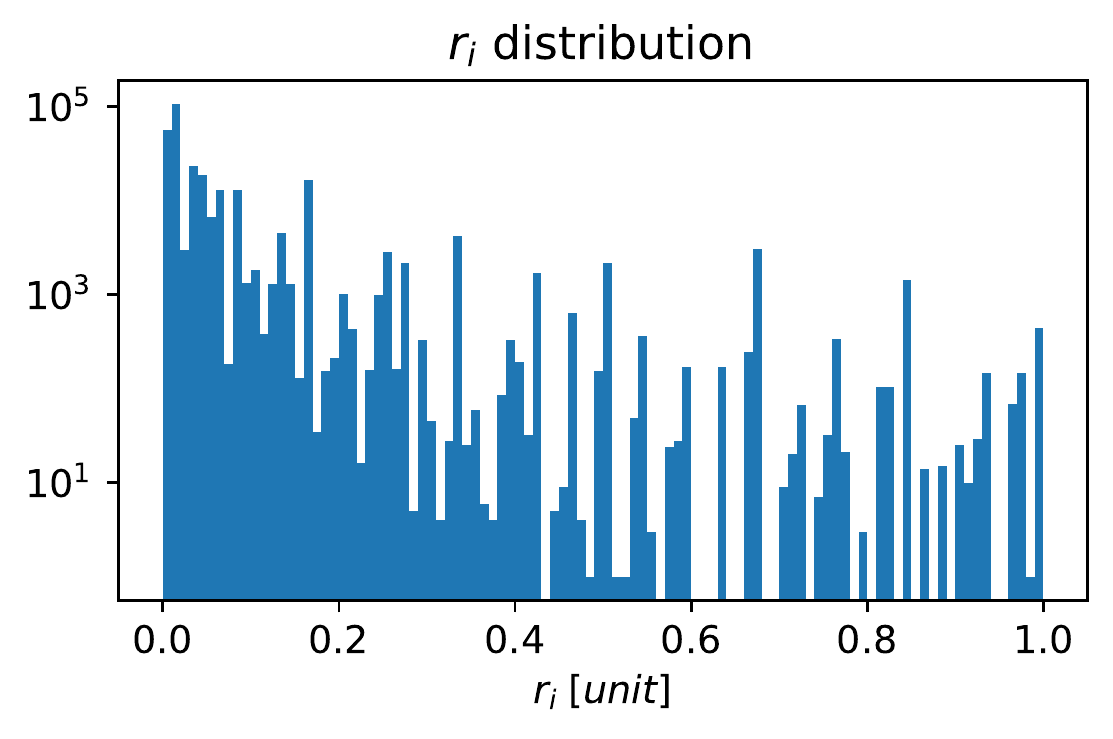}
    }
    \caption{Job processing time and resource requirement distribution in cluster \textit{Zapat} (both axes limited to the 99th percentile).}
    \label{fig:prdist}
\end{figure}

For each combination of the number of jobs $n \in \{500, 1000, 5000, 10000\}$ and the number of machines $m \in \{10, 20, 50, 100\}$, we generated $30$ problem instances. The processing time and resource requirement of each job were set to the processing time and memory requirement of a job randomly chosen from the log. As the lower bound on the optimal schedule length is $L = \max\{\sum_{i \in \jset} p_i/m, \sum_{i \in \jset} p_i\cdot r_i\}$, we only considered instances in which $\max_{i \in \jset} p_i < L.$ This way, we omitted (trivial) instances for which the length of the optimal schedule is determined by a single job.

\subsection{Simulation results}

In Fig.~\ref{fig:approxf}, we present the results obtained for all the algorithms and all the $(n, m)$ combinations.
We report the returned \cmax{} values as normalized by the lower-bound of the optimal schedule length, i.e. $\max\{\sum_{i \in \jset} p_i/m, \sum_{i \in \jset} p_i\cdot r_i\}.$
For lower numbers of machines (left part of the figure) the results of \emph{ApAlg}, \emph{ApAlg-S}, and \emph{ApAlg-H} algorithms are comparable. However, when the number of machines increases, the original approximation algorithm is significantly outperformed by the \emph{ApAlg-S} and \emph{ApAlg-H} variants. In the considered job log, the $50$th percentile on the normalized resource requirement was $0.0169$. We would thus expect that usually not more than $50$ machines are busy in the initial part of a schedule provided by the \emph{ApAlg} variant (due to the threshold of $\frac23$ on a total resource consumption). In such cases, the \emph{ApAlg-S} and \emph{ApAlg-H} variants gain a clear advantage, as they can potentially make use of all the machines, if possible.

\begin{figure}[p]
    \centering
    \subfloat[$(500, 10)$]{\includegraphics[width=0.24\textwidth]{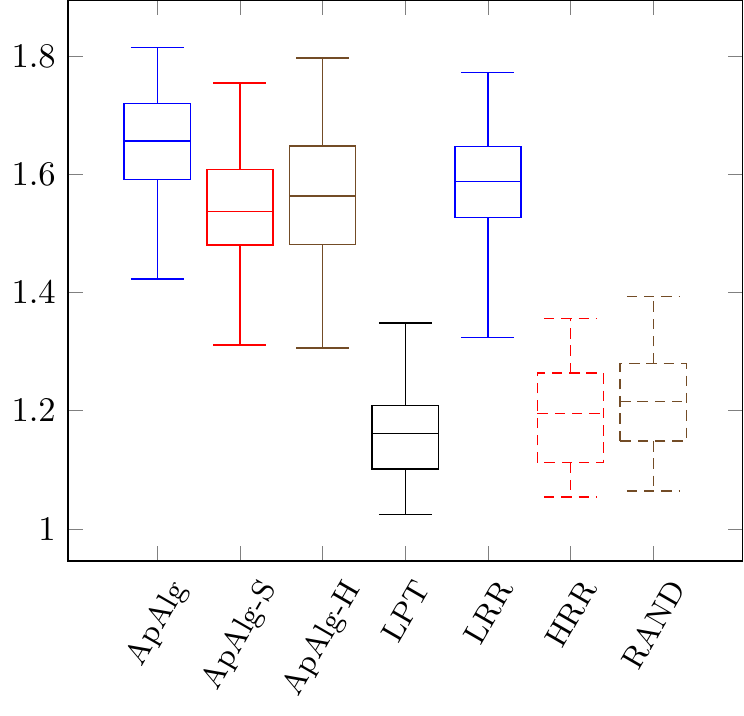}}
    \subfloat[$(500, 20)$]{\includegraphics[width=0.24\textwidth]{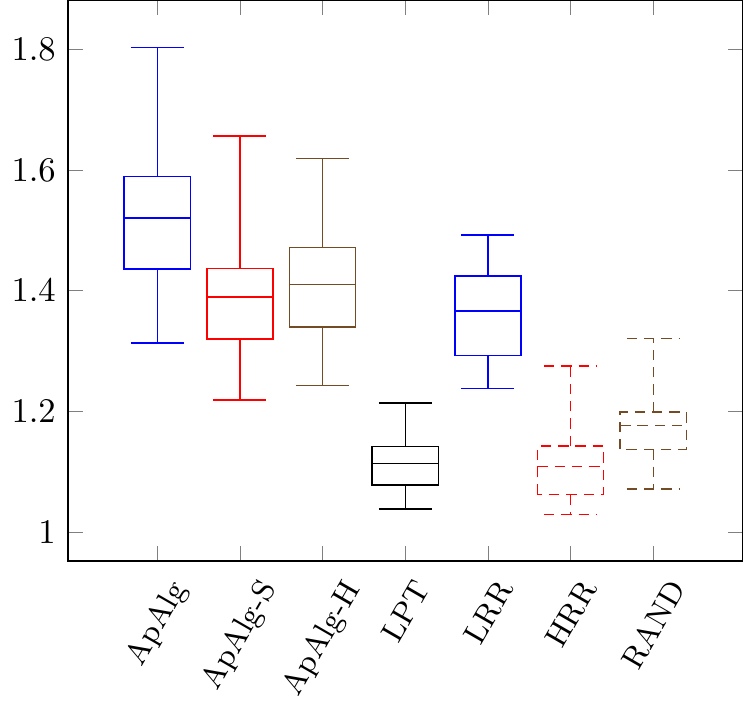}}
    \subfloat[$(500, 50)$]{\includegraphics[width=0.24\textwidth]{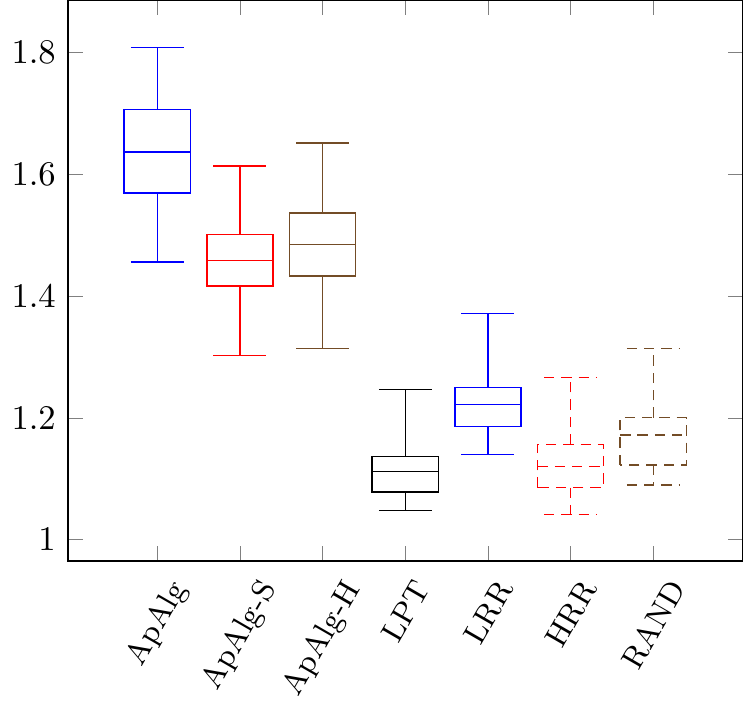}}
    \subfloat[$(500, 100)$]{\includegraphics[width=0.24\textwidth]{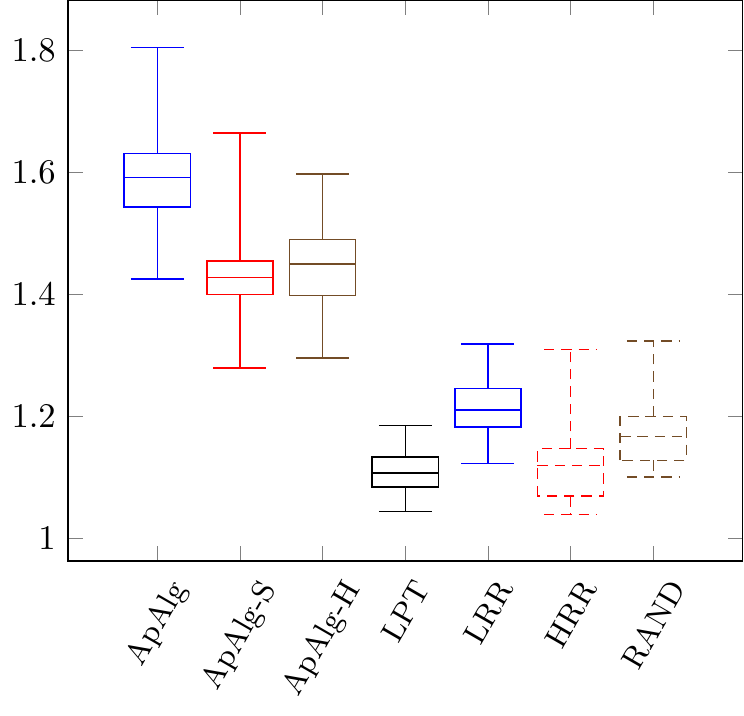}}
    \hfil
    \subfloat[$(1000, 10)$]{\includegraphics[width=0.24\textwidth]{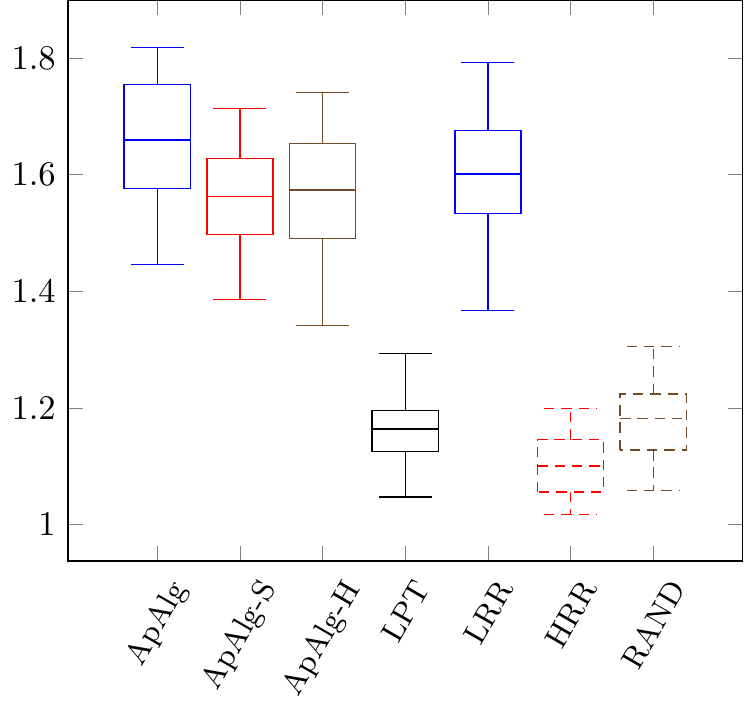}}
    \subfloat[$(1000, 20)$]{\includegraphics[width=0.24\textwidth]{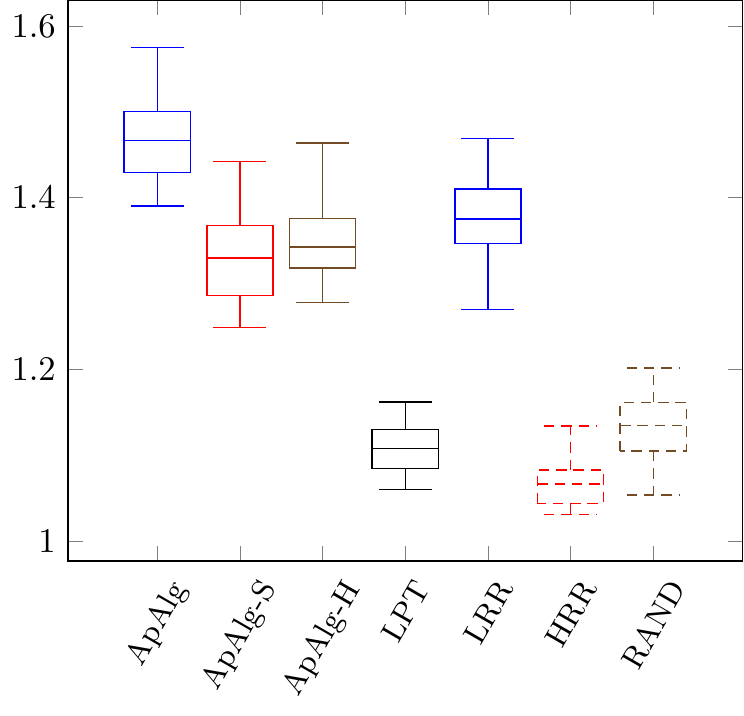}}
    \subfloat[$(1000, 50)$]{\includegraphics[width=0.24\textwidth]{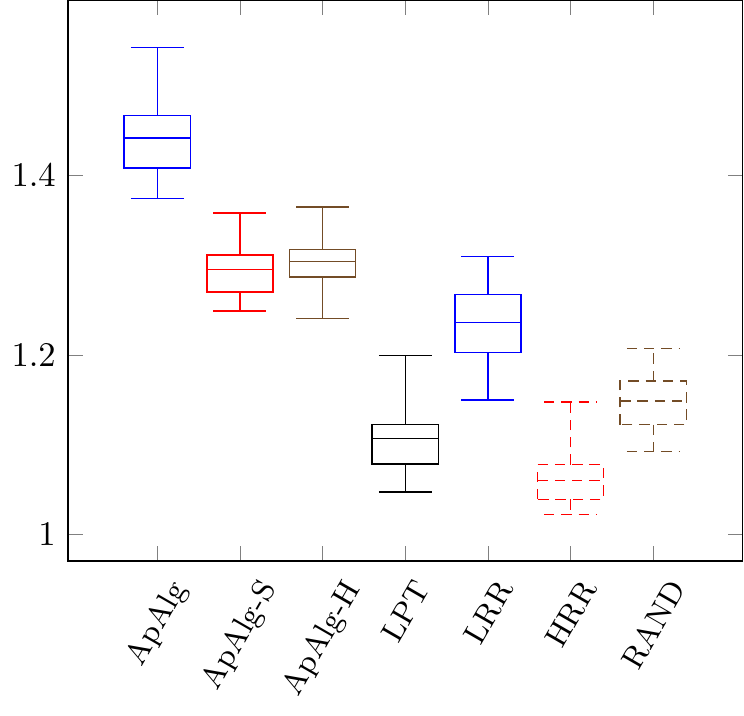}}
    \subfloat[$(1000, 100)$]{\includegraphics[width=0.24\textwidth]{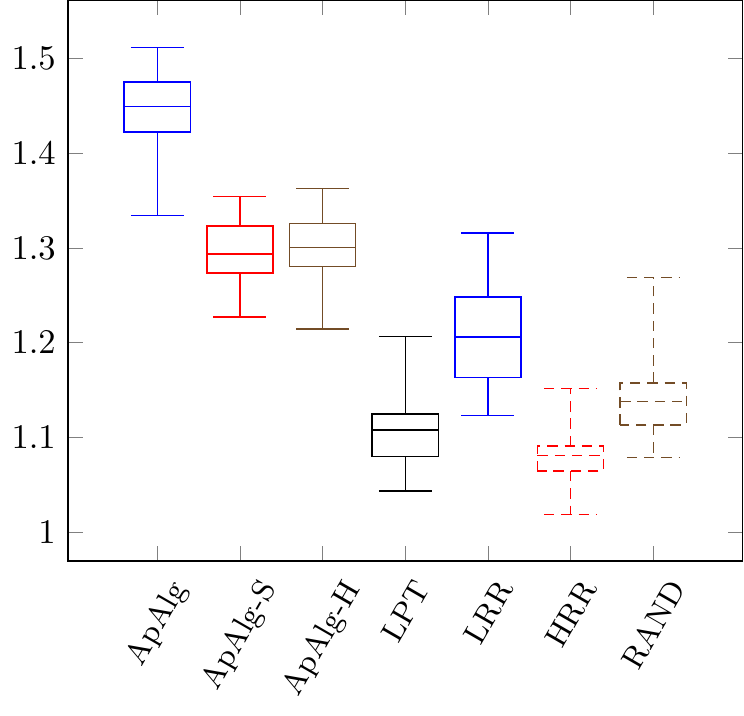}}
    \hfil
    \subfloat[$(5000, 10)$]{\includegraphics[width=0.24\textwidth]{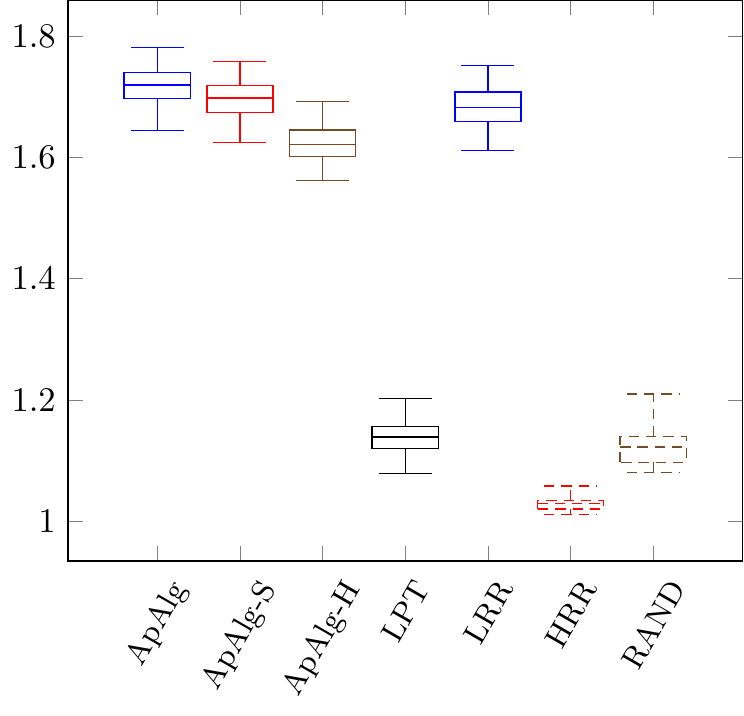}} 
    \subfloat[$(5000, 20)$]{\includegraphics[width=0.24\textwidth]{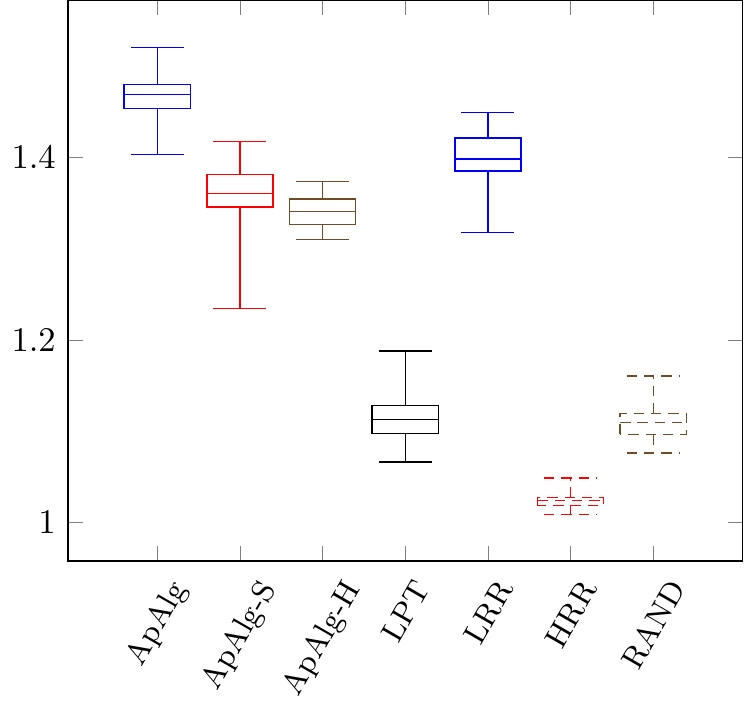}}
    \subfloat[$(5000, 50)$]{\includegraphics[width=0.24\textwidth]{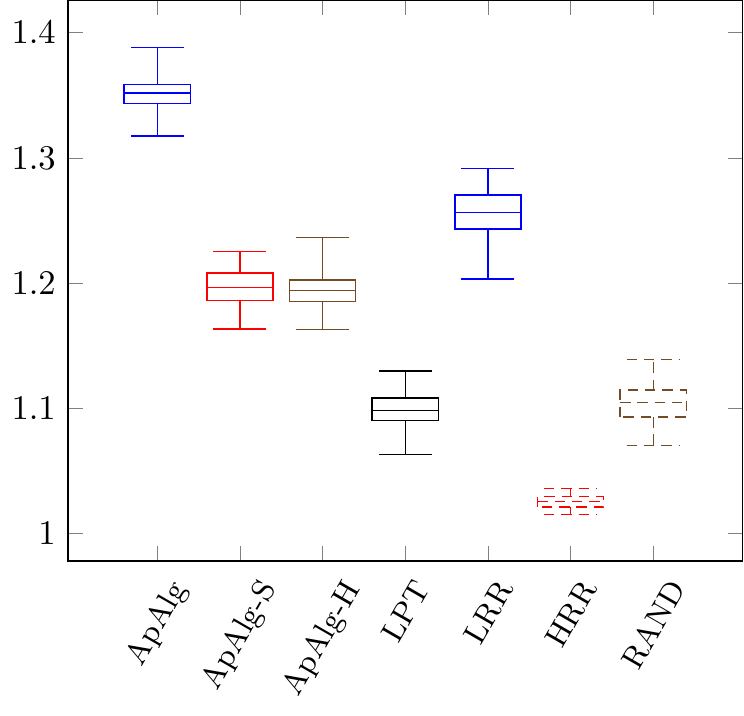}}
    \subfloat[$(5000, 100)$]{\includegraphics[width=0.24\textwidth]{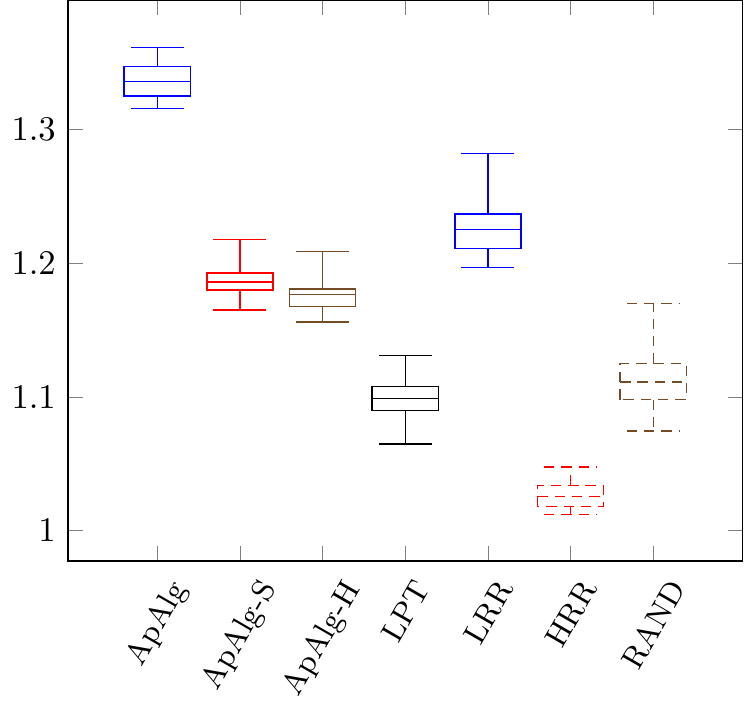}}
    \hfil
    \subfloat[$(10000, 10)$]{\includegraphics[width=0.24\textwidth]{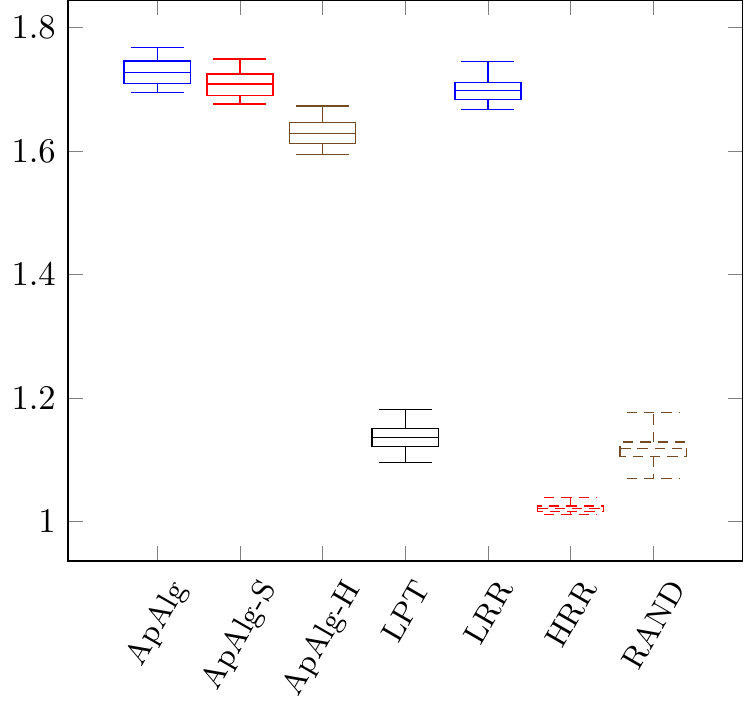}}
    \subfloat[$(10000, 20)$]{\includegraphics[width=0.24\textwidth]{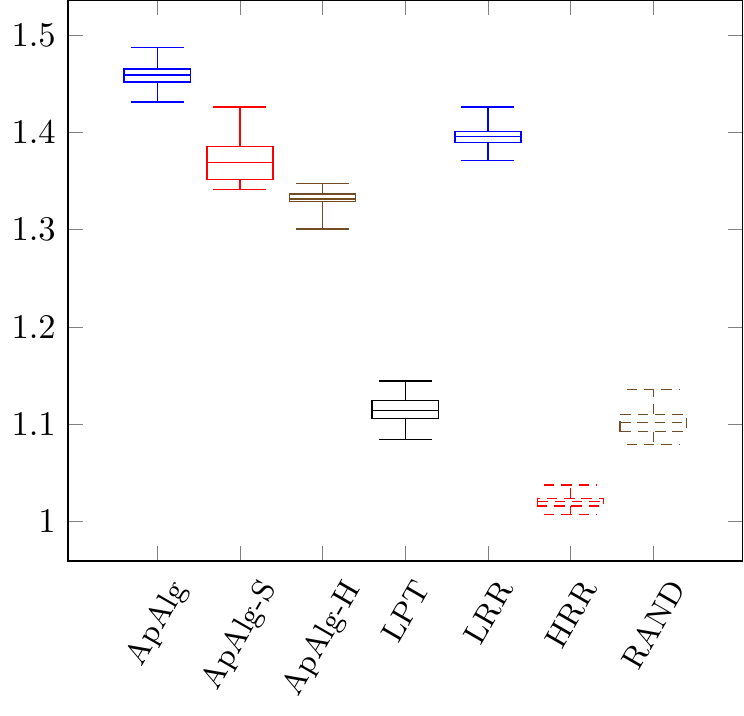}}
    \subfloat[$(10000, 50)$]{\includegraphics[width=0.24\textwidth]{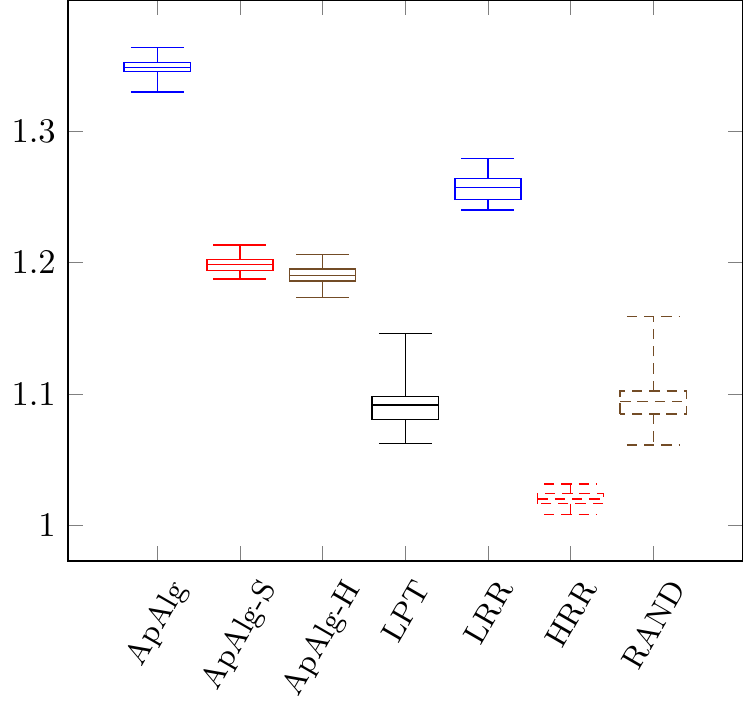}}
    \subfloat[$(10000, 100)$]{\includegraphics[width=0.24\textwidth]{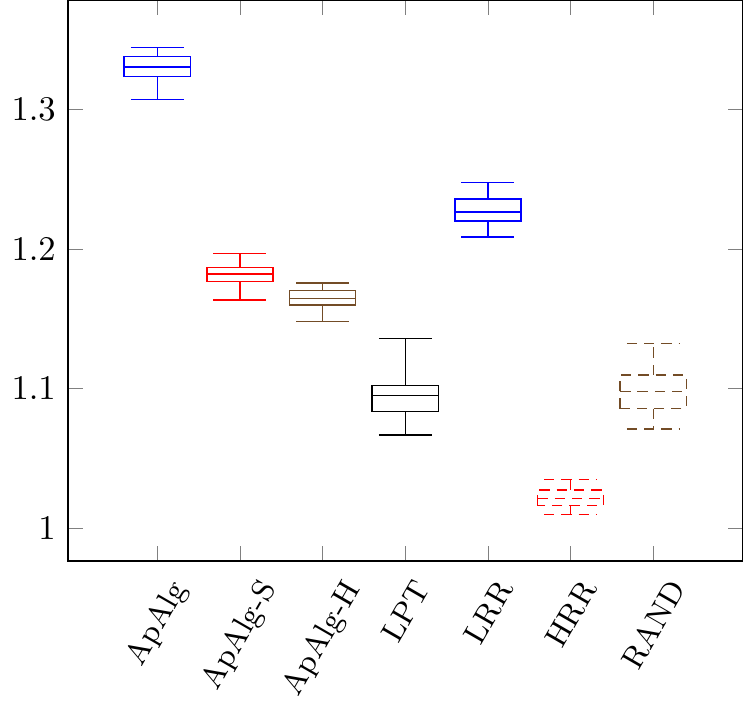}}
    \caption{The \cmax{} values normalized by the lower-bounds on the optimal schedule lengths. Captions $(n,m)$ describe the number of jobs $n$ and the number of machines $m$.}
    \label{fig:approxf}
\end{figure}

When the numbers of jobs and machines increase (right-bottom part of the figure), the normalized \cmax{} values decrease for all the algorithms. In such cases, a greedy heuristic, \emph{HRR}, provides almost optimal schedules, with a maximum normalized \cmax{} value of $1.03$. There might be two reasons for that. First, when the number of jobs increases and their processing times come from the same distribution, it is easier for greedy algorithms to provide better results (as the normalized \cmax{} value is relative). Second, if job resource requirements are not too small compared to the number of machines, the \emph{HRR} produces a schedule with the resource being almost fully-utilized most of the time, compared to \emph{LPT} which makes decisions solely based on the job length.

We also compared the runtime of the algorithms on large instances with $10000$ jobs and $100$ machines. As expected, log-linear \emph{ApAlg}, \emph{ApAlg-H}, \emph{LRR} and \emph{HRR} algorithms were significantly faster: their runtimes (median over $30$ instances) were equal to $2.83$s, $3.29$s, $0.74$s and $0.83$s, respectively --- in contrast to \emph{ApAlg-S}, \emph{LPT} and \emph{RAND} algorithms with runtimes of $88.39$s, $88.58$s and $92.94$s.

\section{Conclusions}
\label{sec:co}

In this paper, we presented a log-linear $2\frac56$-approximation algorithm for the parallel-machine scheduling problem with a single orthogonal resource and makespan as the objective function. Our algorithm improves the $\left(3 - \frac3m\right)$- approximation ratio of Garey and Graham  \cite{Garey1975}. It is also considerably easier to implement than the approximation algorithms proposed by Niemeier and Wiese \cite{NiemeierWiese2013} and Jansen, Maack and Rau \cite{Jansen2019}.

In the computational experiments, we compared three variants of our algorithm to four list scheduling heuristics. We used the real-life data provided by the MetaCentrum Czech National Grid. The results provided by the \emph{HRR} (Highest Resource Requirement) list heuristic significantly outperformed all other algorithms for the considered dataset. Although the results provided by the \emph{HRR} algorithm are promising, the approximation ratio can be improved in general. Thus, in order to provide the best results, one can combine the \emph{HRR} algorithm with our algorithms and thus obtain good schedules with a better guarantee on their approximation ratio.

\section*{Acknowledgements and Data Availability Statement}

This research is supported by a Polish National Science Center grant Opus (UMO-2017/25/B/ST6/00116). The authors would like to thank anonymous reviewers for their in-depth comments that helped to significantly improve the quality of the paper.

The datasets and code generated during and/or analyzed during the current study are available in the Figshare repository: \url{https://doi.org/10.6084/m9.figshare.14748267} \cite{Artifact}.

\bibliographystyle{splncs04}
\bibliography{bibliography}

\end{document}